\documentclass[10pt,twocolumn,twoside]{IEEEtran} 
\usepackage{cite}
\usepackage{multicol}
\usepackage{amsmath}
\usepackage{subeqnarray}
\usepackage{cases}
\usepackage{framed}
\usepackage{amsfonts}
\usepackage{amssymb}
\usepackage{bm}
\usepackage{graphicx}
\usepackage{subfig}
\usepackage[english]{babel}
\usepackage{color}%
\usepackage{algorithm}
\usepackage{url}
\usepackage[belowskip=-6pt,aboveskip=0pt]{caption}
\usepackage{amsthm}
\usepackage{amsbsy}
\newtheorem{thm}{Theorem}
\newtheorem{cor}{Corollary}

\newtheorem{prop}[thm]{Proposition}

\theoremstyle{remark}

\theoremstyle{definition}

\providecommand{\U}[1]{\protect\rule{.1in}{.1in}}

\newcommand{\pdrho}{\frac{\partial}{\partial \rho}}

\begin{document}

\title{\LARGE CoMPflex: CoMP for In-Band Wireless Full Duplex}

\author{Henning Thomsen, Petar Popovski, Elisabeth de Carvalho, Nuno K. Pratas,\\ Dong Min Kim and Federico Boccardi
\thanks{Part of this work has been performed in the framework of the FP7 project ICT-317669 METIS, which is partly funded by the European Union.
Part of this work has been supported by the Danish High Technology Foundation via the Virtuoso project.}
\thanks{H. Thomsen, P. Popovski, E. de Carvalho, N. Pratas and D. M. Kim are with the Department of Electronic Systems, Aalborg University, Denmark
(email: \{ht, petarp, edc, nup, dmk\}@es.aau.dk).}
\thanks{F. Boccardi is with Ofcom, London, United Kingdom
(email:
\newline
federico.boccardi@ofcom.org.uk).
}%
}

%
\maketitle

\vspace{-20mm}

\begin{abstract}
In this letter we consider emulation of a Full Duplex (FD) cellular base station (BS)
by using two spatially separated and coordinated half duplex (HD) BSs. The proposed system is termed \emph{CoMPflex (CoMP for In-Band Wireless Full Duplex)} and at a given instant it serves two HD mobile stations (MSs), one in the uplink and one in the downlink, respectively. 
We evaluate the performance of our
scheme by using a geometric extension of the one-dimensional Wyner model, which takes
into account the distances between the devices. The results show that CoMPflex leads to gains in terms of sum-rate and energy efficiency with respect to the ordinary FD, as well as with respect to a baseline scheme based on unidirectional traffic.
%
\end{abstract}


\vspace{-5mm}

\section{Introduction}\label{sec:Introduction}

5G wireless systems are expected to have a large number of Base Stations (BSs) per unit area~\cite{boccardi2014vodafone}. The inter-BS connections lay the ground to use cooperative techniques, commonly
referred to as \emph{Coordinated Multi-Point (CoMP)}~\cite{lee2012coordinated}. Another
important wireless trend is the use of in-band full-duplex (FD) wireless transceivers  
\cite{DBLP:journals/corr/SabharwalSGBRW13}, expected to lead to a
two-fold spectral efficiency gain~\cite{goyal2013analyzing}. The key challenge in FD
is the high transceiver complexity, required to cope with the
strong self-interference induced by the \emph{downlink (DL)} transmission path into the
\emph{uplink (UL)} receiving path\cite{DBLP:journals/corr/SabharwalSGBRW13}. The high
transceiver complexity will likely make the FD feasible only for the BS, while a Mobile
Station (MS) will remain half-duplex (HD), but the system can still harvest the FD
gain by scheduling at least two different MSs, as shown in Fig.~\ref{fig:FDBS1}.

Traditionally, UL and DL are treated in isolation, i.e. all MSs have only DL or only UL traffic, see Fig.~\ref{fig:MCOMP1}. The advent
of FD has shifted the focus towards two-way optimization of wireless
networks~\cite{InterferenceSpin}.
This leads to two new transmission modes, as the one depicted on Fig.~\ref{fig:HDBS1},
while in the second mode the roles are reversed (MS1 transmits and MS2 receives). In Fig.~\ref{fig:HDBS1} HD-BS1 sends the signal $x_{D1}$ in the DL to MS1,  while HD-BS2 receives $x_{U2}$ in the UL from MS2. HD-BS2 receives interference from HD-BS1, while MS1 receives interference from MS2. However, HD-BS1 can use the high-bandwidth wired connection to HD-BS2 to send the data of $x_{D1}$ to HD-BS2, such that HD-BS2 can recreate $x_{D1}$ and perfectly cancel the interference from HD-BS1. Note that $x_{U2}$
remains as interference to MS1, but it could be mitigated via e.g. scheduling.  The setting on Fig.~\ref{fig:HDBS1} operates equivalently as a single FD-BS, see Fig.~\ref{fig:FDBS1}. With Time Division Duplexing (TDD), the roles of MS1 and MS2 are reversed in the next transmission slot, thus serving both MSs two-way. 

\begin{figure}[tb]
\centering
\subfloat[]{
\includegraphics[width=0.3\linewidth]{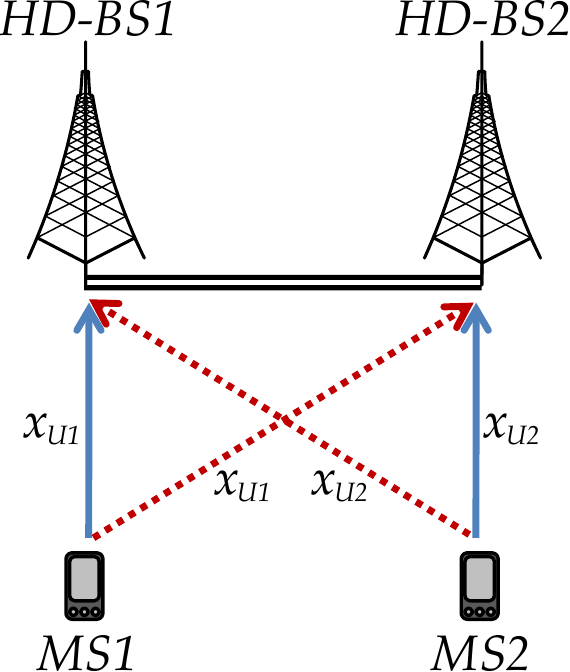}
\label{fig:MCOMP1}}
\subfloat[]{
\includegraphics[width=0.3\linewidth]{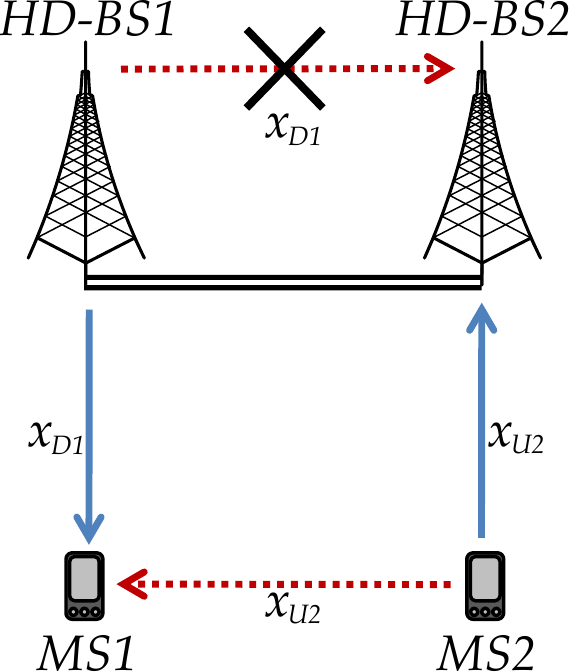}
\label{fig:HDBS1}}
\subfloat[]{
\includegraphics[width=0.3\linewidth]{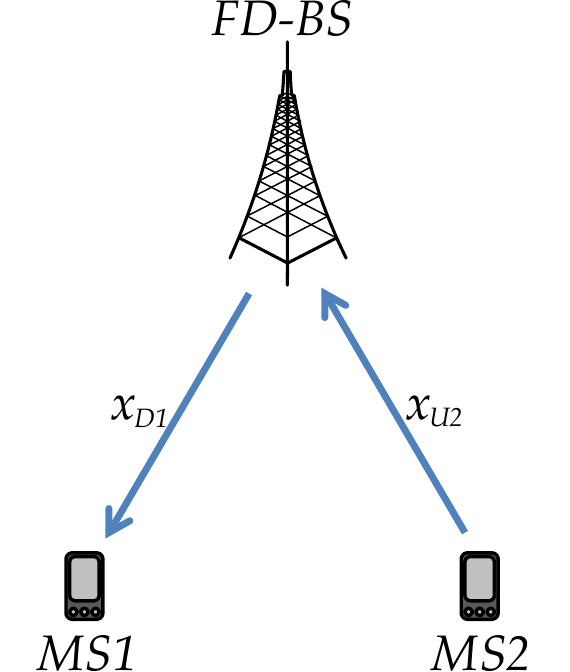}
\label{fig:FDBS1}}
%
\caption{(a) Baseline scheme for UL; (b) The proposed CoMPflex scheme; (c) FD transmission. Solid arrows denote signals, dashed arrows denote interference, the double line is a wired connection.}
\label{fig:Schemes}
\end{figure}

The previous example constitutes the main proposal of this letter: use
the CoMP infrastructure, with interconnected HD-BSs, in order to obtain
a distributed FD implementation, termed \emph{CoMPflex (CoMP for in-band full duplex)}.
In traditional CoMP, the BSs cooperate to serve the MSs either in DL or in UL, but not both simultaneously. The price of the performance gains of traditional CoMP over the non-cooperative schemes is the complex processing of interference using, e.g. dirty paper coding. Different from that, CoMPflex is a way to get performance gains over non-cooperative schemes, without complicated processing at the BSs. Note that the channel between two different BS is static, such that it can be easily estimated.
Assuming a sufficient wired bandwidth between the HD-BSs,
CoMPflex serves simultaneously two terminals that have opposite (UL/DL)
connections. The ordinary FD is a special case of CoMPflex where the HD-BSs are at a distance zero.
CoMPflex brings four advantages: (i) the
coupling losses between the antennas associated with each path are
mitigated; (ii) the self-interference coming from the DL
transmission is now within the same order of magnitude as the UL
signal, reducing the need for a high dynamic range receiver; and (iii)
the reduction of the distance between the cellular devices and HD-BS,
which leads to potential energy savings \cite{chih2014toward} both at
the device and network side, and (iv) the separated transmitter and
receiver of the HD-BSs can bring the infrastructure closer to the
MSs, resulting in rate gains and decrease in interference.
We will show that such a separation is \emph{also} beneficial for a traditional setting with unidirectional traffic. In this way, CoMPflex brings a two-part gain over traditional non-cooperative schemes: (1) by an optimal spatial separation between the HD-BSs; and (2) serving the MSs with bidirectional traffic jointly instead of unidirectional traffic.
Our analysis is based on the classical Wyner model~\cite{wyner1994shannon}, 
enriched to capture the geometric setting of CoMPflex and allow for studying the effect
of varying the distance between the connected BSs.
\begin{figure*}[t]
	\centering
		\includegraphics[width=0.8\linewidth]{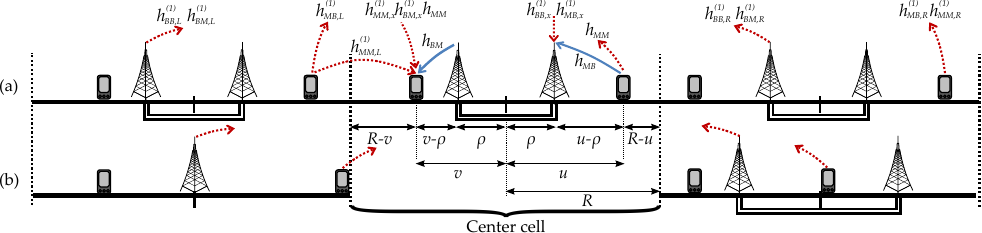}
	\caption{System model for CoMPflex, (a): The center cell and two adjacent cells, one to the left ($L$) and one to the right ($R$), with $x \in \{L, R\}$; (b): The worst case situation of transmitting BSs and MSs from neighboring cells.}
	\label{fig:SystemModelCoMPflex}
\end{figure*}

\vspace{-2mm}

\section{System Model}
\label{sec:SystemModel}

Our system model is perhaps the simplest one where the gain of CoMPflex can be demonstrated. It is based on an one-dimensional Wyner model, as shown in
Fig.~\ref{fig:SystemModelCoMPflex}a, where the cells have radius $R$.\footnote{Here we treat the 1-D scenario. Obviously, the analysis can be extended to 2-D and 3-D models, but this is out of the scope for this introductory article.}
The figure shows a 1-tier interference model, but it can be extended to several interfering tiers.
In the center cell, there is either one FD-BS placed at the cell center (corresponding to $\rho=0$), or two HD-BSs, one for DL placed in the left half of the cell, and one UL placed in the right half (corresponding to $\rho > 0$).
Both HD-BSs have the same distance $\rho$ to the cell center.
We assume a sufficient number of active users in the cell, so that one DL-MS in the left half and one UL-MS in the right half can always be scheduled.
The DL-MS is placed uniformly at random at position $v$ in the left half of the cell, and the UL-MS is placed uniformly at random at position $u$ in the right half of the cell.

In the baseline scheme, the BSs and MSs are deployed as in CoMPflex, except that the traffic is bidirectional in CoMPflex and unidirectional in the baseline. In the baseline scheme the BSs are not cooperating, which is sufficient to show the effect of using bidirectional versus unidirectional traffic. In both schemes, the positions of the nodes in the interfering cells depend on the interference model, see Sec.~\ref{sub:InterferenceRegimes}.

\vspace{-3mm}

\subsection{Signal Model}

We explain the signal model for CoMPflex only; the one for the baseline scheme can be derived from it.
In Fig.~\ref{fig:SystemModelCoMPflex}a, wireless channels are indicated as arrows showing
the transmission direction: a solid line for signal and a dashed line for interference.
For clarity, only the links relevant to the center cell are shown.
All channels are Rayleigh faded with unit mean power.
The pathloss\footnote{An offset equal to 1 is added to avoid a singularity when $d=0$.} is modeled as $\ell(d) = (1+|d|)^{-\alpha}$,
where $d$ is the distance from the transmitter to the receiver and $\alpha$ is the path loss exponent.
The system bandwidth is normalized to 1 Hz. All nodes use the same frequency.
We assume Additive White Gaussian Noise (AWGN), with power $\sigma^2$.
Finally, we assume full channel state information at all nodes.

The distance between BS and UL-MS is $\vert u -
\rho \vert$, and between DL-MS and BS, $\vert v - \rho \vert$. We restrict $\rho$ to 
$\left[0,R \slash 2 \right]$, since $\rho > R \slash 2$ would allow the MS to be closer to an
interfering BS. In the UL, the signal received at the BS, $y_{B}$, is
\begin{align*}
	y_{B} = h_{MB} \ell ( \rho - u)^{\frac{1}{2}} x_M + I_{B} + z_{B}
\end{align*}
where $x_M$ is the symbol sent from MS and $z_{B}$ is the AWGN at the BS. The term
$h_{MB}$ denotes the channel from the MS to the BS. $I_{B} = I_{BB} + I_{MB}$ is the
interference from BSs ($I_{BB}$) and MSs ($I_{MB}$) in the neighboring cells (see
Appendix). Then
\begin{equation}\label{eqn:SINRUL}
	\gamma_{U}(\rho,u) = \frac{P_M(\rho) g_{MB} \ell(\rho - u) }{\sigma^2 + \overline{I}_{B}}
\end{equation}
is the UL SINR, where $P_M = |x_M|^2$, $g_{MB} = \vert h_{MB} \vert^2$ and $\overline{I}_{B} = \vert I_{B} \vert^2$ is the received interference power at the BS.

The DL-MS receives the signal $y_{M}$,
\begin{align*}
	y_{M} = h_{BM} \ell (\rho - v)^{\frac{1}{2}} x_B + h_{MM} \ell ( u + v)^{\frac{1}{2}} x_M + I_{M} + z_{M}
\end{align*}
where $x_B$ is the signal sent from BS, $h_{BM}$ and $h_{MM}$ denote respectively the
channels between BS and MS, and the intra-cell channel between the MSs. $I_{M} = I_{BM}
+ I_{MM}$ is the interference from the BSs ($I_{BM}$) and the users ($I_{MM}$) in
neighboring cells given in the Appendix, and $z_{M}$ is the AWGN at the MS. The DL SINR is
\begin{equation}\label{eqn:SINRDL}
	\gamma_{D}(\rho,u,v) \!=\! \frac{ P_B(\rho) g_{BM} \ell(\rho - v)}{ P_M(\rho) g_{MM} \ell(u+v) \!+\! \overline{I}_{M} \!+\! \sigma^2 }
\end{equation}
where $P_B = |x_B|^2$, $g_{BM} = \vert h_{BM} \vert^2$, $g_{MM} = \vert h_{MM} \vert^2$ and $\overline{I}_{M} = \vert I_{M} \vert^2$ is the received interference power at the MS.


\subsection{Power Adjustment and Inter-Cell Interference Models}
\label{sub:PowerAdjustmentforCoMPflex}
\label{sub:InterferenceRegimes}

The core of CoMPflex is in bringing the MSs and BSs closer to each other, therefore it is natural to allow them to adjust their transmission power accordingly.
This transmission power is computed based on the required power received at either the BS, $P_B^{req}$, and the MS, $P_M^{req}$, when the MS is placed at the cell edge.
The required powers are defined in Sec.~\ref{sec:PerformanceResults}.
We denote the adjusted power as $P_B(\rho)$ and $P_M(\rho)$, for the BS and MS respectively,
which is computed as:
\begin{align*}
	P_y^{req} \leq P_x(\rho) \ell(R-\rho) \Leftrightarrow P_x(\rho) \geq P_y^{req} \ell(R-\rho)^{-1},
\end{align*}
given a required cell edge rate $R_{x_0}$ and a corresponding outage probability $\varepsilon$, where $x,y \in \{B, M\}$. It can be shown that the required power received at $y$, is $P_y^{req} = - \frac{ (2^{R_{x_0}}-1) \sigma^2 }{
\log(1-\varepsilon)\ell(R) }$.
Finally, in Sec.~\ref{sec:PerformanceResults} we will also evaluate the \emph{constant power} case, where $P_B = P_B(\rho=0)$ and $P_M = P_M(\rho=0)$.


We consider two inter-cell interference models when assessing the performance of CoMPflex and the baseline.
%
In the first model, for both CoMPflex and the baseline, nodes in the interfering cells follow the same deployment pattern as the center cell, with one FD-BS or two HD-BSs placed at distance $\rho$ from the center of the cell. The setup is shown in Fig.~\ref{fig:SystemModelCoMPflex}a, and the expression of the interference is shown in the Appendix.

The second model is the worst-case interference, see Fig.~\ref{fig:SystemModelCoMPflex}b, where the transmitting BSs and MSs are indicated by dashed arrows. In each cell, the interfering BSs and MSs are placed as close to the center cell as possible, since we want the interference is maximized at center cell. 
For the cell to the left, the worst-case position of the interfering BS is at $\rho=0$ (the center of its cell) and the MS at the cell edge. For the cell to the right, the worst-case position of the interfering BS is when $\rho=R \slash 2$ (maximal) and the MS is at the cell center. These worst-case scenarios are not formally proved due to lack of space.


\section{Analysis} 
\label{sec:analysis} In this section we show that a nonzero $\rho$ will benefit the
performance in terms of sum-rate and energy efficiency.



The sum-rate $R_{sum}(\rho)$ is given by
\begin{align*}
	R_{sum}(\rho) &= \log_2 \left( 1 +\gamma_{U}(\rho)\right) + \log_2 \left( 1 +\gamma_{D}(\rho)\right) \\ &= \log_2 \left( 1 +\gamma_{U}(\rho) + \gamma_{D}(\rho) + \gamma_{U}(\rho) \cdot \gamma_{D}(\rho)\right).
\end{align*}
%

For analytical tractability we consider \emph{stationary} conditions: (i) the
channel gains of all links are unitary, i.e. $g_x = 1$; (ii) the inter-cell interference
comes from the nearest cell tier; (iii) the interfering MSs are at their
average positions; and (iv) we fix the MSs in the center cell at the positions $u$ and $v$.

\begin{prop}\label{prop:SINRProduct}	
	In stationary conditions the following two statements hold:
	\begin{itemize}
	\item $\forall \rho \in [0, v]$ in the DL, we have $\gamma_{D}(u,v,0) \leq \gamma_{D}(u,v,\rho)$;
	\item $\forall \rho \in [0, u]$ in the UL, we have $\gamma_{U}(u,0) \leq \gamma_{U}(u,\rho)$.
	\end{itemize}
\end{prop}
\begin{proof}

We only sketch the proof of the first statement; the second one is similar.
The proof is done by analyzing the derivative of the DL SINR, which is denoted $\gamma_{D}(u,v,\rho) = \frac{S(v,\rho)}{I(u,v,\rho)+z_{M}}$, with respect to $\rho$. Here the DL signal is $S(v,\rho)$ and the DL interference is $I(u,v,\rho)$.
Then from the rule of differentiating a quotient, we only look at the numerator of this derivative, since the denominator is squared.
The DL signal is $S(v,\rho) = P_B(\rho)\ell(R-\rho)$ and the DL interference is
\begin{align*}
	I(u,v,\rho) &= (1+R-\rho)^{\alpha} \cdot \Big[ P_M^{req} \ell(2R+\rho-v) \nonumber \\
	& \hspace{-1cm}  + P_M^{req} \ell(2R+\rho-v) + P_B^{req} \ell( 1 + 3R \slash 2 - v ) \nonumber \\
	& \hspace{-1cm} + P_B^{req} \ell( 1 + 5R \slash 2 + v ) + P_B^{req} \ell(u+v) \Big] \nonumber \\
	& \hspace{-1cm} = (1+R-\rho)^{\alpha} \cdot I(u,v,\rho).
\end{align*}
%
The numerator of $\pdrho \gamma(\rho)$ equals
\begin{align*}
	&\pdrho S(v,\rho)(I(u,v,\rho)+z_{M}) - \pdrho (I(u,v,\rho)+z_{M})S(v, \rho).
\end{align*}
For $\rho \in [0, v]$, the derivative os $S(v,\rho)$ is
\begin{align*}	
	\pdrho S(v,\rho) & = P_M^{req}\alpha(1+R-\rho)^{\alpha-1}(1+v-\rho)^{-\alpha} \nonumber \\  & \cdot \left( -1 + (1+R-\rho)(1+v-\rho)^{-1} \right),
\end{align*}
which is positive since $-1 + (1+R-\rho)(1+v-\rho)^{-1} > 0$.
The derivative $\pdrho I(u,v,\rho)$ equals
\begin{align*}
	&\pdrho (1+R-\rho)^{\alpha} I(u,v,\rho) + (1+R-\rho)^{\alpha} \pdrho I(u,v,\rho) \\
	& -\alpha (1+R-\rho)^{\alpha-1} I(u,v,\rho) + (1+R-\rho)^{\alpha} \pdrho I(u,v,\rho).
\end{align*}

It can be shown that $\pdrho I(u,v,\rho) < 0$. Using the fact that $S(v,\rho) > 0$,
$\pdrho S(v,\rho) > 0$ and $I(u,v,\rho)+z_{MS} > 0$, we can conclude that $\pdrho
\gamma_D(\rho) > 0, \forall \rho \in [0, v]$
\end{proof}
%
\begin{cor}\label{cor:SINRProduct}
	In stationary conditions, $\forall \rho \in [0, \min\{u,v\}]$,
	\begin{equation*}\label{eqn:objectiveFunction}
		\gamma_{U}(0) \cdot \gamma_{D}(0) \leq \gamma_{U}(\rho) \cdot \gamma_{D}(\rho).
	\end{equation*}
\end{cor}
From Prop.~\ref{prop:SINRProduct} and Cor.~\ref{cor:SINRProduct}, we can conclude that the sum-rate $R_{sum}(\rho)$, increases for $\rho \in [0, \min\{u,v\}]$.

The energy efficiency $EE$ is defined as the amount of bits transmitted per unit of energy~\cite{li2011energy}. It is given by
\begin{equation}\label{eqn:EEdef}
	EE(\rho) = \frac{R_{sum}(\rho)}{P_B(\rho) + P_M(\rho)}.
\end{equation}
CoMPflex improves the EE, both due to increased sum-rate and lower transmission power as functions of $\rho$.
%





\section{Performance Results}\label{sec:PerformanceResults}

We evaluate the sum-rate and EE of CoMPflex via numerical simulations.
The main simulation settings are listed in Tab.~\ref{tab:SimulationParameters}.
The positions of the MSs are random and uniform over their half-cell, all links have fading and we include $N$ interfering cell tiers both to the left and right.
\begin{table}[h]
	\caption{Simulation parameters.}
	\centering
			\begin{tabular}{ c l c }
		  	\hline
  			Parameter & Description & Simulation Setting\\
				\hline
				$R$ & Cell radius & $100$ m \\
				$\sigma^2$ & Noise power at MS and BS & $-174$ dBm \\
				$\alpha$ & Path loss exponent & $3,4,5$ \\
				$N$ & Number of interfering cells & $10$ \\
				$R_{U_0}$ & Required UL rate & $0.03$ bps \cite[Ch.11]{holma2012lte} \\
				$R_{D_0}$ & Required DL rate & $0.06$ bps \cite[Ch.11]{holma2012lte}\\
				$\varepsilon$ & Cell edge outage probability & $0.1$ \\
				\hline
			\end{tabular}
	\label{tab:SimulationParameters}
\end{table}
%

The sum-rate performance results are shown in Fig.~\ref{fig:SumRateComparisons}. We compare the proposed CoMPflex scheme with the baseline, using both deployments depicted in Figs.\ref{fig:SystemModelCoMPflex}a and \ref{fig:SystemModelCoMPflex}b, the latter denoted as the worst-case interference. Since the baseline serves unidirectional traffic over the entire transmission phase, we need two such phases to serve two-way traffic to the MSs. Therefore, the sum-rate is $\frac{R_U+R_D}{2}$, where $R_U$ is the total UL rate and $R_D$ is the total DL rate. We also depict the 
results for $\rho=0$ (no splitting) for CoMPflex (ordinary FD) and the baseline. The performance of CoMPflex is almost always better than the baseline. Also, in both interference deployments, there is a benefit in increasing $\rho$. This confirms the conclusions from Sec.~\ref{sec:analysis}, but now with more realistic conditions. We also see that the difference between the sum-rate when using constant power versus power adjustment is negligible, such that power adjustment is chiefly beneficial for higher EE. It is also interesting that there is a gain in sum-rate \emph{both} from splitting the BSs, as well as using the CoMPflex scheme instead of the baseline, even in the worst case.

\begin{figure}[t]
	\centering
		\includegraphics[width=0.85\linewidth]{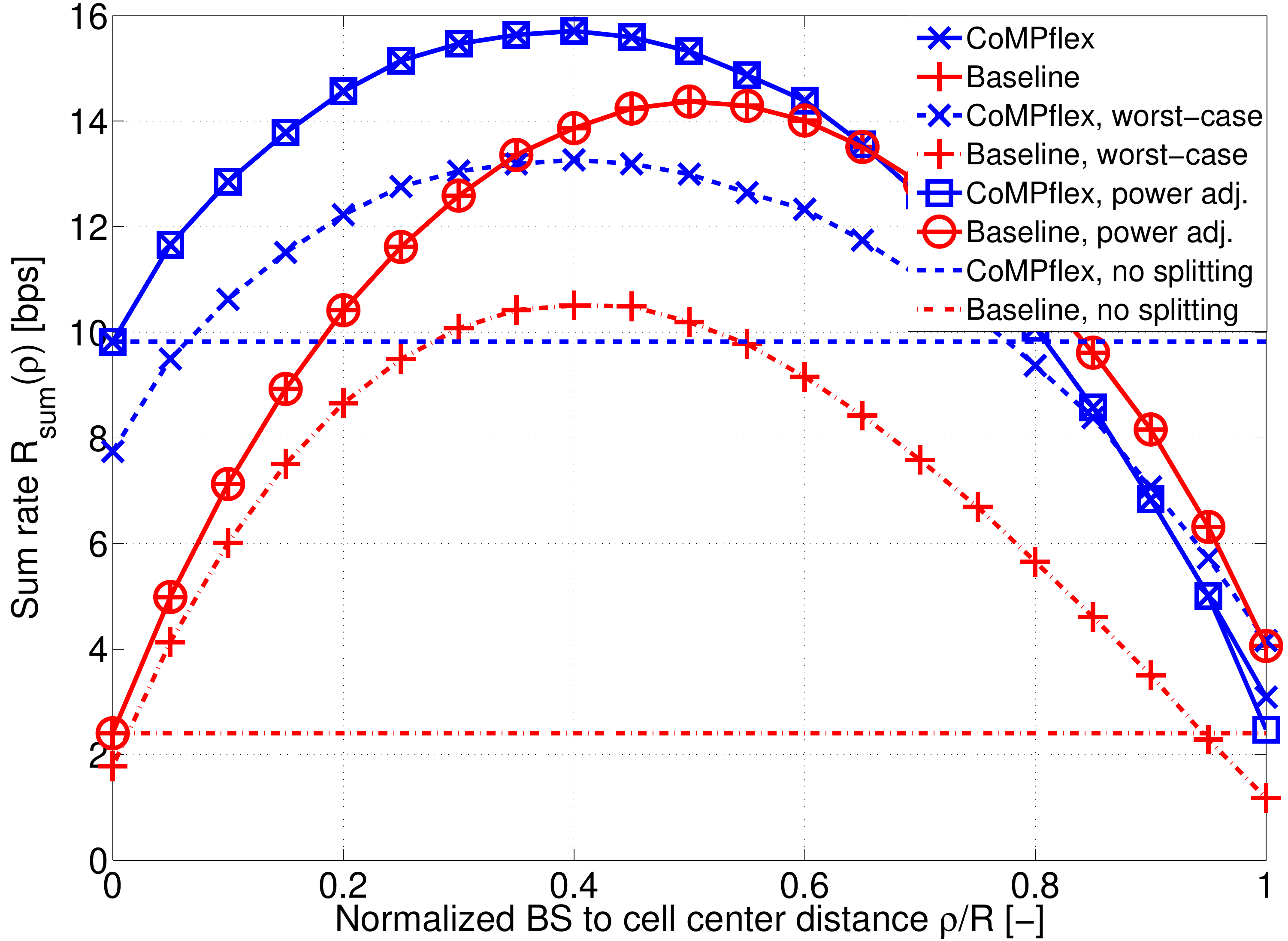}
	\caption{Sum-rate comparison of CoMPflex and baseline, comparing worst-case with simulation, where $\alpha = 4$.}
	\label{fig:SumRateComparisons}
\end{figure}
%
%
%
To evaluate the EE gains of CoMPflex, we normalize the achieved EE at a certain $\rho$ with the EE when $\rho = 0$ (Fig.\ref{fig:Schemes}c).
The normalized EE is:
\begin{equation}\label{eqn:normalizedEE}
	\eta(\rho) = \frac{EE(\rho)}{EE(\rho = 0)}.
\end{equation}
We also look at the total transmission power $P_{sum}(\rho) = P_B(\rho)+P_M(\rho)$, and
study how it depends on $\rho$. These curves, along with the normalized EE, are shown in
Fig.~\ref{fig:EEandSumPowerCombined}, where we see that the total transmission power decreases as $\rho$ increases.
This is to be expected since the distance between a MS and its serving BS is decreased, requiring less power is used to maintain the same performance.
The combination of a lower required transmission power with a higher sum-rate results in the observed dramatic increase in EE, especially when $\rho$ tends to $R/2$.
This increase is further accentuated in propagation environments with higher $\alpha$ values.

\section{Conclusion}\label{sec:Conclusion}

We have shown that the proposed CoMPflex scheme allows emulation of in-band FD operation by spatially dislocating the UL and DL traffic into two HD-BSs.
Our results show that this splitting leads to a substantial increase in sum-rate as well as an EE increase ranging from $15\times$ to $45\times$ for high BSs splitting distance and increasing path-loss exponents. This initial study uses a simple model, but the gains and the insights obtained warrant further analysis that relies on more complex models, such as the ones based on stochastic geometry.

\appendix
\label{subsec:InterferenceExpressionsAll}

Using Fig.~\ref{fig:SystemModelCoMPflex}, the distances between the MS in the $n$th cell and the receiving BS are
\begin{align*}
	d_{MB,L}^{(n)} = 2nR - u_{n,L} + \rho, &\hspace{0.75cm} d_{MB,R}^{(n)} = 2nR + u_{n,R} - \rho,
\end{align*}
for the left (L) and right (R) cell respectively. $u_{n,L}$ and $u_{n,R}$ describe the same distances as $u$ and $v$ in the left and right interfering cells.
The BS-MS, BS-BS and MS-MS distances are, respectively,
\begin{align*}
	d_{BM,L}^{(n)} = 2nR + \rho - v,&\hspace{0.75cm} d_{BM,R}^{(n)} = 2nR - \rho + v,\\
	d_{BB,L}^{(n)} = 2nR + 2\rho, &\hspace{0.75cm} d_{BB,R}^{(n)} = 2nR - 2\rho,\\
	d_{MM,L}^{(n)} = 2nR - u_{n,L} - v, &\hspace{0.75cm} d_{MM,R}^{(n)} = 2nR + u_{n,R} + v,
\end{align*}
where the random variables $u$, $v$, $u_{n,L}$ and $u_{n,R}$ are independent. In the
worst-case scenario (Fig.~\ref{fig:SystemModelCoMPflex}b), the MS-BS, BS-MS,
BS-BS and MS-MS distances are, respectively,
\begin{align*}
	d_{MB,L}^{(n),wc} = 2nR - R + \rho, &\hspace{0.75cm} d_{MB,R}^{(n),wc} = 2nR - \rho,\\
	d_{BM,L}^{(n),wc} = 2nR - v,&\hspace{0.75cm} d_{BM,R}^{(n),wc} = 2nR - R \slash 2 + v,\\
	d_{BB,L}^{(n),wc} = 2nR + \rho, &\hspace{0.75cm} d_{BB,R}^{(n),wc} = 2nR - R \slash 2 - \rho,\\
	d_{MM,L}^{(n),wc} = 2nR - R - v, &\hspace{0.75cm} d_{MM,R}^{(n),wc} = 2nR + v.
\end{align*}
Letting $\psi \in \{ MB, BM, BB, MM \}$, and $h_{\psi,\varphi}^{(n)}$,
$\varphi \in \{L,R\}$ the channel between a node in the $n$th cell and
the receiver, the interference terms are then
\begin{equation}
	I_{\psi} = \sum_{n=1}^{\infty} \left[ h_{\psi,L}^{(n)} \sqrt{\ell (d_{\psi,L}^{(n)})} x_{L}^{(n)} + h_{\psi,R}^{(n)} \sqrt{\ell (d_{\psi,R}^{(n)})} x_{R}^{(n)} \right],
\end{equation}
where $x_{L}^{(n)}$ and $x_{R}^{(n)}$ are symbols sent from appropriate nodes.
\begin{figure}[t]
	\centering
		\includegraphics[width=0.825\linewidth]{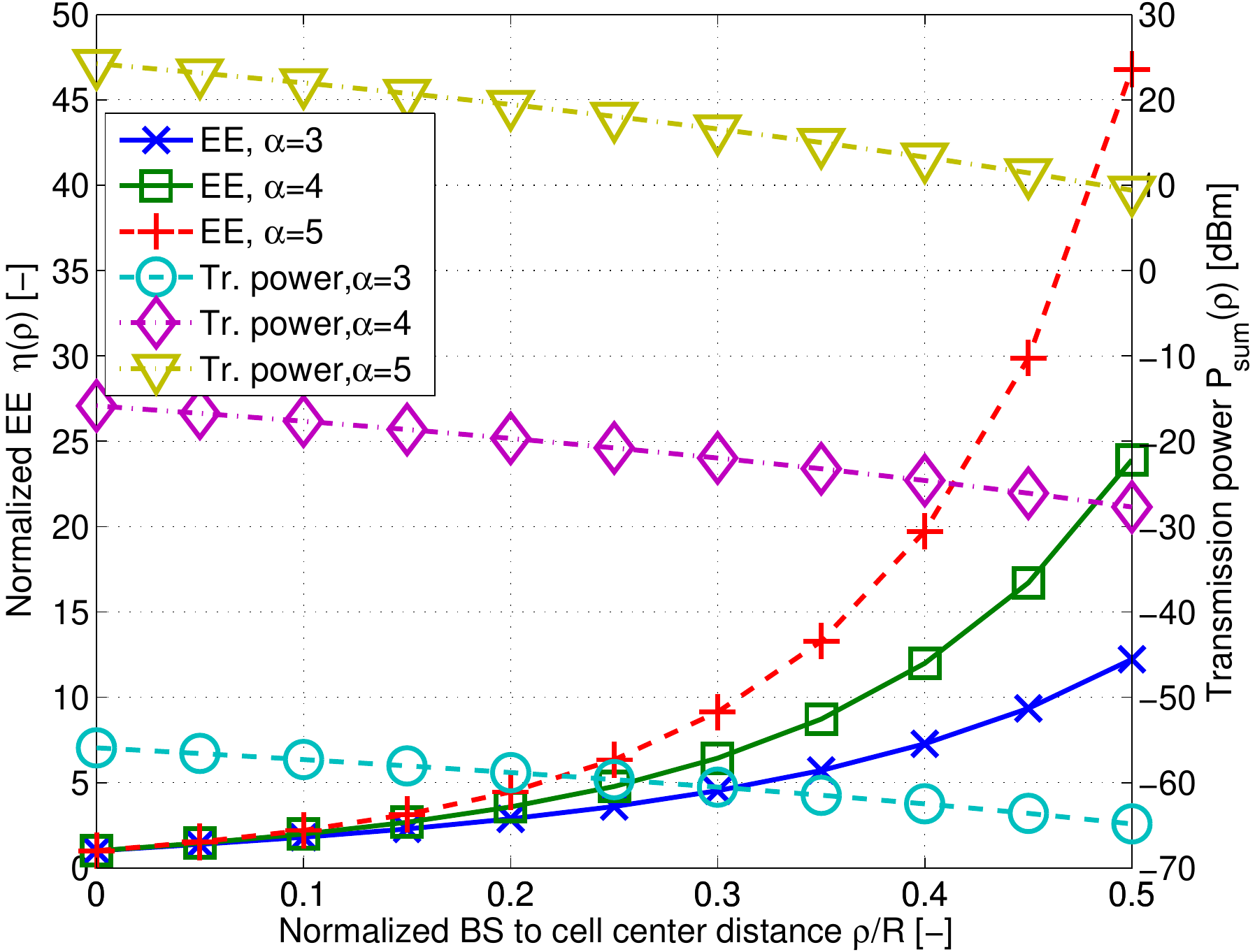}
	\caption{Power usage and EE of CoMPflex.}
	\label{fig:EEandSumPowerCombined}
\end{figure}

\bibliographystyle{ieeetr}
\bibliography{biblio}

\end{document}